\documentclass[runningheads,a4paper]{llncs}
\usepackage{amssymb}
\usepackage{amsmath}
\usepackage{algorithm}
\usepackage{algorithmic}
\usepackage{verbatim}
\usepackage{caption}
\usepackage{amsmath}
\usepackage{paralist}
\usepackage{url}

\newcommand{\keywords}[1]{\par\addvspace\baselineskip
\noindent\keywordname\enspace\ignorespaces#1}

\begin{document}

\title{Bin Packing with Multiple Colors}

\author{Hamza Alsarhan, Davin Chia, Ananya Christman\thanks{Corresponding author}, Shannia Fu, \and Yanfeng Jin}

\institute{Middlebury College, Middlebury VT 05753\\
\email{\{halsarhan, dchia, achristman, sfu, yjin\}}@middlebury.edu}

\maketitle

\medskip

\begin{abstract}
\emph{
In the Colored Bin Packing problem a set of items with varying weights and colors must be packed into bins of uniform weight limit such that no two items of the same color may be packed adjacently within a bin. We solve this problem for the case where there are two or more colors when the items have zero weight and when the items have unit weight. Our algorithms are optimal and run in linear time. Since our algorithms apply for two or more colors, they demonstrate that the problem does not get harder as the number of colors increases. We also provide closed-form expressions for the optimal number of bins.
}
\end{abstract}\section{Introduction}

In the \textit{Black and White Bin Packing} problem proposed by Balogh et al.~\cite{jbalogh}, we are given a set of items each with a weight and a color - either white or black. We are given an unlimited number of bins each with a weight limit and we assume that no item weighs more than the bin weight limit. Formally, the input is a set $S$ of $n$ black and white items with weights $w_1, \ldots, w_n$ and bins of weight limit $C$ where each $w_i < C$. The goal is to pack the items in as few bins as possible while maintaining that no two items of the same color are packed adjacently within a bin. We extend the work of~\cite{jbalogh} by considering more than 2 colors. Specifically, we consider the offline case where all items are given in advance and there is no ordering among them. We first solve the problem for the case where the items have zero weight. We present a linear time algorithm, \textsc{Alternate-Zero}, that solves this problem optimally. We then present another linear time algorithm,  \textsc{Alternate-Unit}, that optimally solves the problem for unit-weight items.

This \textit{Colored Bin Packing} problem has many practical applications. For example, television and radio stations often assign a set of programs to their channels. Each program may fall into a genre such as  comedy, documentary, and sports, on TV, or various musical genres on radio. To maintain a diverse schedule of programs, the station would like to avoid broadcasting two
programs of the same genre one after the other. This problem can be modeled as Colored Bin Packing where the items correspond to programs, colors to genres and bins to channels.

Another application involves generating diverse content to be displayed on websites. Many websites prefer displays that alternate between different types of information and advertisements and the Colored Bin Packing problem can be used to generate such a display. The items correspond to the contents to display, the colors to the type of content, and the bins to the size of the display.

The remainder of this paper is organized as follows. In Section~\ref{related} we discuss several results related to the Colored Bin Packing problem. In Section~\ref{zero} we present our results for zero-weight items. In Section~\ref{unit} we present our main result - a solution for the problem for unit-weight items. In Section~\ref{conc} we summarize our results and discuss possible extensions.

\subsection{Related Work}

\label{related}

Classic Bin Packing is a well known problem, with many applications throughout various fields. Bin Packing with color constraints is relatively recent, arising with the increasing complexity of contemporary industrial processes. The use of color gives classic Bin Packing added flexibility in modelling real-world problems. For example, Oh and Son used color, specifically, the requirement that two items with the same color cannot be packed in the same bin, to efficiently assign tasks in real time within a multiprocessor system~\cite{yoh}. Their result was a modified First Fit algorithm that is 1.7-competitive in the worst case, and 1.1 in the average case. Dawande et al. investigated a version of Colored Bin Packing where each bin has a maximum \textit{color capacity}~\cite{mdawande} i.e. a limit on the number of items of a particular color. Such a problem models the slab design problem in the production planning process of a steel plant. The authors present two 3-approximation algorithms, classical First-Fit-Decreasing, and a modified First Fit. Xavier and Miyazawa use Bin Packing with Colors (they refer to the colors as \textit{classes}) to model Video-on-Demand applications~\cite{exavier}. Specifically, there is a server of multiple disks that each have limited storage capacity and can hold video files of varying sizes and genres (i.e. classes). The application is given an expected number of requests for movies based on movie/genre popularity. The goal is to construct a server that maximizes the number of satisfied requests.

Epstein et al. study a version of Colored Bin Packing where items have weights and colors and bins have weight and color constraints, specifically no more than $k$ items of the same color may be packed in a bin~\cite{lepstein}. They prove upper and lower bounds for several variants of this problem in both the offline and online settings. A more generalized version of constrained bin packing where the constraints among items are defined in a conflict graph have been studied by Jansen~\cite{kjansen} and Muritiba et al.~\cite{amuritiba}.

In this paper, we investigate Colored Bin Packing with \textit{alternation} constraints - i.e. two items with the same color cannot be packed adjacently to each other in a bin. This 2-color variant was first introduced by Balogh et al., who studied both the offline and online versions of the problem~\cite{jbalogh}. For the offline version the authors presented a 2.5-approximation algorithm. For the online version, they proved that classical algorithms (First, Best, Worst, Next, Harmonic) are not constant competitive, and hence do not perform well. They further proved a universal lower bound of approximately 1.7213. This exceeds classical Bin Packing's upper bound of 1.5889, proving that Colored Bin Packing with alternation is harder~\cite{sseidens}. The authors' main result was a 3-competitive online algorithm, \textit{Pseudo}. The main concept of the algorithm is quite elegant. Using ``pseudo", i.e. dummy, items of different color and zero weight, their algorithm creates an acceptable sequence of colors that are later substituted with actually released items. If any bin overflows, Pseudo redistributes the extra items using Any Fit into available bins, and creates new bins if necessary.

Bohm et al. explored a zero-weight version of the 2-color variant of this problem~\cite{mbohm}. They show that in this case, the optimal number of bins equals the color \textit{discrepancy} - the absolute difference between the number of items of the two colors. The paper's main result is an optimal algorithm, named Balancing Any Fit, which operates on the simple principle that an item should be packed in an available bin containing the most oppositely-colored items. This is made easier since weight is not an issue.

Dosa and Epstein~\cite{gdosa} extend the results from~\cite{jbalogh}. They prove that online Colored Bin Packing with alternation for 3 or more colors is harder than the 2-color version. Furthermore, the authors show that 2-color algorithms do not apply to 3 or more color problems. Lastly, they present a 4-competitive algorithm for the problem. This is a modified version of the earlier Pseudo, appropriately named Balanced-Pseudo. Instead of randomly assigning items to available bins, Balanced-Pseudo assigns an item to the bin whose top color is the most frequent.

It was proven in~\cite{gdosa}, that no competitive algorithm exists for online Colored Bin Packing when there are more than two colors even when the items have zero-weight. Therefore, in this paper we consider offline Colored Bin Packing when there are more than two colors. We present optimal linear time algorithms that solve this problem for two or more colors when the items have zero-weight and when the items have unit-weight. Since our algorithms are optimal for two or more colors, they demonstrate that the problem does not get harder as the number of colors increases.

\section{Zero-Weight Colored Bin Packing}

\label{zero}

In the zero-weight case, the only constraint is that no items of the same color may be adjacent in a bin. Since the items have no weight, the bins do not have a weight limit (formally, $L = n$). Thus, we pay special attention to the most frequent color, i.e. the color with the greatest number of items, since this color will determine the final number of bins used.  We refer to this color as $MaxColor$ and the set of all non-$MaxColor$ colors as $OtherColors$. We let $MaxCount$ and $OtherCount$ denote the number of items of $MaxColor$ and $OtherColors$, respectively, that have not yet been packed.

Our general approach is to alternate between $MaxColor$ items and $OtherColors$ items such that we pack as many $MaxColor$ items as possible. This technique allows us to pack as many $MaxColor$ items in a bin as possible, and thereby minimizes the number of bins we use.

\begin{algorithm}
\caption{\textsc{alternate-zero}($S$). $S$ is a set of $n$ items.}
\label{zeroalg}
\begin{algorithmic} [1]
\STATE $MaxColor$: most frequent color
\STATE $MaxCount$: number of items of $MaxColor$
\STATE $OtherColors$: set of all non-$MaxColor$ colors
\STATE $OtherCount$: number of items of any $OtherColors$.
\STATE $Discrepancy = MaxCount - OtherCount$
\IF {$Discrepancy \le 0$}
\STATE Alternate between items of $OtherColors$ until there is one fewer $OtherColors$ item than $MaxColor$ item remaining.
 \STATE Start with a $MaxColor$ item and alternate between items of $OtherColors$ and $MaxColor$ until all items are packed. Only one bin is needed.
\ELSE
\STATE Alternate packing an item of $MaxColor$ and an item of any arbitrary $OtherColors$ until all $OtherColors$ items are packed. This requires one bin.
\STATE There will be $MaxCount - OtherCount - 1$ remaining $MaxColor$ items so use this many number of bins to pack each of these items. The total number of bins is $Discrepancy$.
\ENDIF
\end{algorithmic}

\end{algorithm}

As in~\cite{mbohm}, we describe the number of required bins in terms of the \textit{discrepancy} of the items, i.e. the difference between the number of $MaxColor$ and $OtherColors$ items that have not already been packed. For example, if there are 6 black items, 3 white items and 2 red items that must be packed, the discrepancy is 6-(3+2)=1. It was shown in~\cite{mbohm} that for the case where there are exactly two colors, the discrepancy is a lower bound for the number of required bins. In particular, the minimum number of bins is achieved by packing the first bin by starting with a $MaxColor$ item, then alternating between an $OtherColors$ item and a $MaxColor$ item until there are no $OtherColors$ items left, then packing one more $MaxColor$ item. At this point there will be $MaxCount-OtherCount-1$ remaining $MaxColor$ items so this many more bins are needed, for a total of $MaxCount - OtherCount$ bins, i.e. the discrepancy. Our linear-time algorithm \textsc{Alternate-Zero} solves the version of the problem where there are more than two colors (see Algorithm~\ref{zeroalg}).

\begin{theorem}
The Algorithm \textsc{alternate-zero} is optimal for the Zero-Weight Colored Bin Packing Problem.
\end{theorem}

\begin{proof}

We consider two cases based on $Discrepancy$:

\medskip

\noindent \underline{Case 1}: $Discrepancy \le 0$.
In this case there are at least as many $OtherColors$ items as there are $MaxColor$ items. Therefore, we need only one bin, which is the fewest number of bins any algorithm will use. To do this, we alternate between items of $OtherColors$ until $OtherCount$ is one less than $MaxCount$. We then add one $MaxColor$ item to the bin. From here, we alternate between any $OtherColors$ item and a $MaxColor$ item until all items are packed. The last item will be one of $MaxColor$.

Note it is always possible to alternate the $OtherColors$ items until $OtherCount$ is one less than $MaxCount$. Suppose, by contradiction, that when we stop alternating, $OtherCount$ is equal or greater than $MaxCount$. Since we alternated between differently colored $OtherColors$ items, we can no longer alternate only if $OtherColors$ now contains just a single color, $X$. By definition of alternating, we know that the bin must contain at least one $X$-colored item. As we have not packed any $MaxColor$ items yet, there are still $MaxCount$  $MaxColor$ items remaining. This means there are at least $MaxCount + 1$ $X$-colored items, which contradicts our original assumption that $MaxColor$ is the color with the maximum number of items. Thus, we can always alternate $OtherColors$ items until $OtherCount$ is one less than $MaxCount$.

Thus, after we fill the first bin we have one more $MaxColor$ items than all $OtherColors$ items remaining. It is easy to see that we can alternate between these items in the same bin. Therefore, we use one bin.

As an example, if we are given 3 White (W), 2 Black (B), 2 Yellow (Y), and 1 Red (R) items, then $MaxColor$ is White, $MaxCount = 3$, $OtherCount = 5$, so $Discrepancy = -2$. One optimal packing is: BYRWBWYW

\medskip

\noindent \underline{Case 2}: $Discrepancy > 0 $. In this case, there are enough $MaxColor$ items that each $OtherColors$ item is needed to be packed in between two $MaxColor$ items. Therefore, all $OtherColors$ items can be regarded as a single color not equal to $MaxColor$, making this the black/white bin packing problem~\cite{jbalogh}. The first bin will start with a $MaxColor$ item, alternate between $MaxColor$ and $OtherColors$ items, and end with a $MaxColor$ item. Thus, the first bin will contain all the $OtherColors$ items and $OtherCount +1$ $MaxColor$ items. There will be $MaxCount - OtherCount -1$ $MaxColor$ items remaining, each of which must be packed in a separate bin. This gives a total of $MaxCount - OtherCount$ bins, i.e. the discrepancy. Since the first bin contains as many $MaxColor$ items as possible and none of the additional bins can be combined, this is the optimal packing.
\end{proof}

As an example, if we are given 8 W, 2 B, and 2 Y items, then $MaxCount = 8$, $OtherCount = 4$, so $Discrepancy = 4$. One optimal packing is: WBWBWYWYW / W / W / W ( where / denotes a new bin).

\section{Unit-Weight Colored Bin Packing}

\label{unit}

\begin{theorem}
The Algorithm \textsc{alternate-unit} is optimal for the Unit-Weight Colored Bin Packing Problem.
\end{theorem}

Algorithm \textsc{alternate-unit} solves the problem when the items have unit weight and each bin has a weight limit $L$. The general idea is as follows. The unit-weight case introduces weight constraints. However, color constraints, specifically the discrepancy, may force us to use more bins than the weight constraint alone. Therefore, the first stage of our algorithm depends on discrepancy. Specifically, if discrepancy does not pose a problem (i.e. is not more than zero), we simply pack the items as in the zero-weight case (using one bin), then split the items into bins of weight equal to the weight limit. If \textit{discrepancy} is an issue (i.e. is more than zero), how we pack depends on whether $L$ is even or odd. In both cases, we pack bins by alternating between $MaxColor$ and $OtherColor$ items. If $L$ is even, we may be able to reduce the total number of bins by using the $OtherColors$ items that top some bins to combine bins that contain only a single $MaxColor$ item. Since all bins are optimally packed, this yields an optimal packing. If $L$ is odd, every full bin will contain one more $MaxColor$ item than $OtherColors$ item, so we check if our packing eventually reduces discrepancy to zero. If so, we pack as we did in the zero-weight case. Otherwise, we will have bins containing only a single $MaxColor$ item. Since these bins cannot be combined and all other bins are optimally packed, the final packing is optimal. Algorithm~\ref{unitalg} describes the steps in detail.

\begin{proof}
For the unit-weight case, we must deal with both color and weight constraints. As in the zero-weight case, we consider two cases based on the discrepancy:

\medskip

\noindent \underline{Case 1}: $Discrepancy \le 0$.
In this case there are enough $OtherColors$ items to pack in between all $MaxColor$ items, so the color constraints do not pose a problem. We know from the zero-weight case that the \textsc{Alternate-Zero} algorithm satisfies the color constraint when discrepancy is less than or equal to 0. Therefore, we call \textsc{Alternate-Zero} to pack the items. Since \textsc{Alternate-Zero} does not deal with weights, it will pack the items into one bin according to color. We then simply split this bin into bins of weight limit $L$. Since \textsc{Alternate-Zero} is optimal, we will use the optimal number of bins.

As an example, if we are given 4 W, 3 B, and 2 Y items and $L = 3$, then $MaxCount = 4$, $OtherCount = 5$, so $Discrepancy = -1$. One optimal packing is BYW / BWB / WYW.

\smallskip

\noindent \underline{Case 2}: $Discrepancy > 0$.
This is the complex case as now the packing depends on both the weight and color constraints. We pack bins by starting with a $MaxColor$ item and alternating between $OtherColors$ and $MaxColor$ items until the bin is full. Since the discrepancy is more than 0, there may be $MaxColor$ items remaining. The remainder of the packing depends on whether $L$ is even or odd.

We first note that since there are more $MaxColor$ items than $OtherColors$ items, packing as many $MaxColor$ items as possible into a bin can never yield a sub-optimal solution. This is because the only reason a $MaxColor$ item would be ``saved" for another bin is if it is needed for packing between two $OtherColors$ items. However, this would never happen if there are more $MaxColor$ items than $OtherColors$ items. We use this fact to prove the optimality of our algorithms for the case where discrepancy is more than zero.

\begin{itemize}

\item \underline{Case(2a)}: $L$ is even.

We refer to the packing before the call to \textsc{Combine} (see Algorithm~\ref{unitalg}) as the \textit{initial packing}. After the initial packing, there will be \textit{M-bins}, \textit{F-bins}, and possibly one \textit{P-bin} (we use the same definitions of these bins as stated in \textsc{Combine}).

\newpage

\begin{algorithm}\caption{ \textsc{Alternate-Unit}($S$, $L$). Input is a set $S$ of $n$ items and bin weight limit $L$.}
\begin{algorithmic}[!htbp]
\label{unitalg}
\STATE $MaxColor$: most frequent color
\STATE $MaxCount$: number of items of $MaxColor$
\STATE $OtherColors$: set of all non-$MaxColor$ colors
\STATE $OtherCount$: number of items of any $OtherColors$
\STATE $Discrepancy = MaxCount - OtherCount$
\STATE $Solution = \emptyset $

\IF {$Discrepancy \le 0$}
	   \STATE $Solution \leftarrow ${\textsc{Alternate-Zero}}($S$).
       \STATE Split $Solution$ into $\lceil\frac{n}{L}\rceil$ bins, each containing $L$ items (possibly fewer for the last bin if there are not enough items to fill it).
	   \STATE \textbf{return} $Solution$
\ELSE
	   \IF {$L$ is even}
            \WHILE {there are $OtherColors$ items remaining}
		          \STATE {Open a new bin $C$. Start with packing a $MaxColor$ item, alternating with an $OtherColors$ item (without overfilling $C$). Add $C$ to $Solution$.}
            \ENDWHILE
		    \STATE {Pack each of the remaining $MaxColor$ items into a separate bin. Add these bins to $Solution$.}
		    \STATE {Call  \textsc{Combine($Solution$)} to combine all the single-item bins. }
            \STATE {\textbf{return} $Solution$}
	   \ELSIF {$L$ is odd}
            \IF {$Discrepancy \leqslant   \lceil \frac{OTHER}{\lfloor{L/2}\rfloor} \rceil $}
                    \STATE //there are enough $OtherColors$ to reduce $Discrepancy$ to 0
				    \WHILE {$Discrepancy > 0$ }
					       \STATE {Open a new bin $C$. Start with packing a $MaxColor$ item, alternating with an $OtherColors$ item (without overfilling $C$). Add $C$ to $Solution$.}
				    \ENDWHILE
                    \STATE {Let $R$ denote the set of remaining items. }
                    \STATE {$Remaining \leftarrow $ \textsc{Alternate-Zero}($R$) //apply the \textsc{Alternate-Zero} algorithm while satisfying the weight constraints.}
                    \STATE {Split $Remaining$ into bins containing $L$ items each (possibly fewer for the last bin if there are not enough items to fill it).}
                    \STATE {Add $Remaining$ to $Solution$}.
	               \STATE {\textbf{return} $Solution$}
			\ELSE
				\WHILE {there are still non-packed $OtherColors$ items:}
                    \STATE {Open a new current bin $C$. Start with packing a $MaxColor$ item, alternating with an $OtherColors$ item, and topping with a $MaxColor$ item (without overfilling $C$). Add $C$ to $Solution$.}
				    \STATE {Pack each of the remaining $MaxColor$ items in its own bin. Add these bins to $Solution$.}
                    \STATE {\textbf{return} $Solution$}
                \ENDWHILE
			\ENDIF
	   \ENDIF
\ENDIF
\end{algorithmic}
\end{algorithm}

\begin{algorithm}\caption{\textsc{Combine} ($Bins$). $Bins$ is a set of bins.}
\begin{algorithmic}[!htbp]

	\STATE Identify the following sets of bins:
\STATE M-bins (MaxColor bins): contain only one item and the item is of $MaxColor$.
\STATE P-bins (Partially full bins): contain both $MaxColor$ and $OtherColors$ items, are topped with a $MaxColor$ item, and can fit at least two more items. There will be at most one of this bin type.
\STATE F-bins (Full bins): fully packed and are topped with an $OtherColors$ item.
\IF {P-bins is not empty}
	   \STATE Let $current$ be the bin in P-bins.
\ELSIF {M-bins is not empty}
      \STATE Let $current$ be a bin in M-bins.
      \STATE Remove $current$ from M-bins.
\ENDIF
\WHILE {F-bins is not empty and M-bins is not empty}
		\IF {weight of $current$ $+ 2 > L$}
			\STATE Remove $current$ from its containing set.
            \STATE Let $current$ be the next bin in M-bins.  //switch to a new bin since $current$ is full
		\ENDIF
		\STATE Let $x$ be the top-most $OtherColors$ item from a bin $F$ in F-bins. Remove $F$ from F-bins.
        \STATE Let $y$ be a $MaxColor$ item from a bin $M$ in M-bins. Remove $M$ from M-bins.
        \STATE Remove $x$ and $y$ from their containing bins.
		\STATE Pack $x$ followed by $y$ into $current$.
	\ENDWHILE
\end{algorithmic}
\end{algorithm}

\newpage

M-bins contain only one item and the item is a $MaxColor$ item. F-bins are completely full so each will start with a $MaxColor$ item, end with an $OtherColors$ item, and contain an equal number of $MaxColor$ and $OtherColors$ items. The sub-routine \textsc{Combine} uses these $OtherColors$ items to combine the M-bins and reduce the total number of bins. If we run out of $OtherColors$ items while packing a bin, the bin may become a P-bin. Since we pack as many $MaxColor$ items as possible, we top the bin with a $MaxColor$ item. If the bin is able to accept two or more items after we do so, it will be a P-bin. Since this occurs only when we run out of $OtherColors$ items, at most one P-bin may exist.

As an example, if we are given 15 W, 4 B, 3 Y, 3 G items and $L=6$, then $MaxCount = 15$, $OtherCount = 10$ so $Discrepancy = 5$.
We initially pack these items as WBWBWB / WBWYWY / WYWGWG / WGW / W / W / W / W.

\textsc{Combine} uses the $OtherColors$ items that top the full bins to reduce the number of bins to WBWBW / WBWYW / WYWGW / WGWBW / WYWGW.

\medskip

To see why the algorithm is optimal, note that after the initial packing, each of the full bins contain as many $MaxColor$ items as possible and are therefore optimally packed. Therefore, if \textsc{Combine} optimally packs the non-full bins, then the entire solution must be optimal. We prove that \textsc{Combine} optimally packs the non-full bins by showing that the number of bins in the solution cannot be further reduced after {\textsc{Combine}} is finished executing.

There are three conditions that can cause \textsc{Combine} to finish executing:
\smallskip
\begin{enumerate}
\item The solution contains no F-bins.\\

\item The solution contains no M-bins.\\

\item The solution contains neither F-bins nor M-bins.
\end{enumerate}

\medskip

Since combining bins requires both F-bins and M-bins, any one of these three conditions means that we cannot reduce the total number of bins by combining bins.

\medskip

\item \underline{Case(2b)}: $L$ is odd.

We begin by packing bins full by alternating between $MaxColor$ and $OtherColors$ items. We refer to this as the \textit{initial packing}. When we end the initial packing depends on how much the discrepancy reduces as items are packed. To see this, first note that when $L$ is odd, a bin that is packed full by alternating between $MaxColor$ and $OtherColors$ will contain one more $MaxColor$ item than $OtherColors$ item. Therefore, when $L$ is odd, the overall number of bins depends on whether the initial discrepancy is low enough that this packing will eventually reduce the discrepancy to zero. We let $D$ denote the discrepancy. Note that after packing one bin in this way, $D$ reduces by 1, after packing two bins in this way, $D$ reduces by 2, and so on. Therefore, $D$ reduces to zero if we are able to pack $D$ bins. From this we see that if $D \leqslant \lceil \frac{OtherCount}{\lfloor{L/2}\rfloor} \rceil $, then we can pack $D$ bins and therefore eventually reduce $D$ to zero. If this condition holds, then we end the initial packing after we've packed $D$ bins. Since $D$ reduces to zero, we now call \textsc{Alternate-Zero} on the remaining items. Since \textsc{Alternate-Zero} does not deal with weights, it will pack the remaining items into one bin according to color. We then simply split this bin into bins of weight limit $L$.

If $D$ cannot be eventually reduced to zero, then there will be $MaxColor$ items remaining after the initial packing. Each of these items must be packed in its own bin (i.e. an M-bin).

Note that since the initial packing maximally packs the $MaxColor$ items, the bins packed during the initial packing are optimal. The rest of the packing depends on $D$ so we prove the optimality of each case separately:

\begin{itemize}

\item[$\circ$]\underline{Case (2bi)}: $D \leqslant \lceil \frac{OtherCount}{\lfloor{L/2}\rfloor} \rceil $, so $D$ eventually reduces to 0.

Since we know \textsc{Alternate-Zero} is optimal when the discrepancy is 0, if the first $D$ bins are optimally packed, the entire packing must be optimal. We show this by considering the number of P-bins that the initial packing yields. First note that the initial packing may yield at most one P-bin. This is because if there were multiple P-bins, then items from one P-bin could be packed into another, thereby turning the latter into an F-bin. Therefore there may be either one or no P-bins after the initial packing.

\begin{itemize}
\item[$\bullet$] {Initial packing yields a P-bin.}

As an example, if we are given 11 W, 2 B, 2 Y, and 3 G and $L=5$, then $MaxCount = 11$, $OtherCount = 7$ so $D = 4$.

The initial packing yields: WBWBW / WBWBW / YWYWG / WGW

No bins can be combined so this is the optimal solution.

\medskip

If the initial packing yields a P-bin, it means that $D = 1$ when we started filling the P-bin. To see why, note that the $D = MaxCount - OtherCount = 0$ when the initial packing is done.

\smallskip

The initial packing ends when we pack the last excess $MaxColor$ item into the P-bin. This reduces $D$ to 0. Since the initial packing packs each bin full and maximizes the number of $MaxColor$ items per bin, it yields an optimal packing for the first $D$ items packed. Therefore, if we can show that there will be no more items remaining after the initial packing, then we know this packing is final and optimal.

\smallskip

Since $D = 0$, the number of $MaxColor$ and $OtherColors$ items must be equal. Any additional $OtherColors$ item would mean the P-bin would not exist, as the last item packed into the bin would be an $OtherColors$ item. Since there are no $OtherColors$ items remaining there must be no $MaxColor$ items remaining, and therefore no items are remaining.

\smallskip

\item[$\bullet$] {Initial packing does not yield any P-bins.}

As an example, if we are given 10 W, 4 B, 2 Y, and 2 G and $L=5$, then $MaxCount = 10$, $OtherCount = 8$ so $D = 2$.

The initial packing yields 2 bins: WBWBW / WBWBW

The discrepancy of remaining items is now zero so we pack them using \textsc{Alternate-Zero} as: YWYWG / WGW

So the final packing is: WBWBW / WBWBW / YWYWG / WGW





\medskip

We know that the bins packed during the initial packing are optimally packed, so we must show that the remaining items will be optimally packed. There are two sub-cases:

\begin{enumerate}
\item $MaxCount = OtherCount = 0$:

Clearly, there are no remaining items to pack so the packing is optimal.

\item $MaxCount = OtherCount > 0$:

Since there are no excess $MaxColor$ items, there will be no M-bins to combine so the packing is optimal.

 \end{enumerate}

\end{itemize}

\smallskip

 \item[$\circ$] \underline{Case 2bii)}: {$D > \lceil \frac{OtherCount}{\lfloor{L/2}\rfloor} \rceil$ so $D$ cannot be reduced to 0.}

As an example, if we are given 15 W, 3 B, 2 Y, and 2 G items and $L=5$, then $MaxCount = 15$, $OtherCount = 7$ so $D = 8$.

The initial packing yields: WBWBW / WBWYW / WYWGW / WGW / W / W / W / W
No bins can be combined so the solution is optimal.

\medskip

When $D$ cannot be reduced to zero, it means that there were not enough $OtherColors$ items to pack the excess $MaxColor$ items. Since the solution contains bins that are maximally packed with $MaxColor$ items and M-bins, no bins can be combined, so the solution is optimal.
\end{itemize}

\end{itemize}

\end{proof}

\subsection{Time Complexity}
We analyze the time complexity based on the various cases.
If discrepancy is less than or equal to 0, the algorithm runs \textsc{Alternate-Zero} to order the items based on color, then packs the items based on weight, using $O(n)$ time in total (where $n$ is the number of items).

If discrepancy is more than 0, we consider two sub-cases: $L$ is odd or $L$ is even.
If $L$ is odd, we check if we can eventually reduce the discrepancy to 0. If we can, we pack items until $D$ is 0. This takes at most $O(n)$, since we consider each item only
once. When $D$ is zero we call \textsc{Alternate-Zero} which takes $O(n)$ time.  If we cannot
reduce discrepancy to 0, we simply pack items by alternating colors. Since each item is considered exactly once, this takes $O(n)$ time.

If $L$ is even, again we pack by alternating colors then call \textsc{Combine}. The runtime of \textsc{Combine} depends on the
number of bins and since there will be $O(n)$ bins, \textsc{Combine} runs in $O(n)$ time.

\subsection{Number of Bins}

In this section, we provide closed-form expressions for the optimal number of bins. The zero-weight case is straightforward: if the discrepancy is no more than 0, then one bin is needed; otherwise discrepancy number of bins are needed.

We now provide these expressions for the unit-weight case by considering four sub-cases:

\begin{enumerate}

\item $L$ is even and $Disrepancy > 0$:

Since $Discrepency$ is positive, we may eventually combine bins. The final number of bins depends on the initial packing: specifically the number of fully packed bins (F-bins), partially packed bins (P-bins), if any, and bins containing a single $MaxColor$ item (M-bins) - these counts will indicate how many M-bins can be combined to reduce the final number of bins.

\smallskip

After the initial packing, every F-bin will contain $L/2$ $OtherColors$ items so there will be $F = \left \lfloor{OtherCount/(L/2)}\right \rfloor$ F-bins after the initial packing.

\smallskip


After the initial packing there will be at most one P-bin. Whether the P-bin exists depends on how many $OtherColors$ items remain after all the F-bins are packed. If no $OtherColors$ items remain, no P-bin will exist. The number of $OtherColors$ items remaining is $R = \left \lfloor{OtherCount\%(L/2)}\right \rfloor$. If $R >$ 0, there will be a partially full bin containing $R$ $OtherColors$ items. Since we start and end this bin with a $MaxColor$ item, this bin must contain $R+1$ $MaxColor$ items. If the bin cannot fit at least two more items then it is a partially full bin that is not a P-bin (i.e. it will not get any additional items in the \textsc{condense} step). If the bin has space for at least two more items, then it is a P-bin and in the \textsc{condense} step, we can pack an equal number of additional $OtherColors$  and $MaxColor$ items in the bin until there are $L-1$ items in the bin. Therefore we can pack an additional $P = (L-1 - R - R - 1)/2$ $MaxColor$ items and $P$ $OtherColors$ items in this bin. If the P-bin does not exist, then $P = 0$.

\smallskip

There are $F(L/2)$ $MaxColor$ items packed in F-bins and $R+1$ $MaxColor$ items packed in the P-bin (if it exists). Therefore, after the initial packing the number of M-bins will be $M = MaxCount - F(L/2) - R-1$.
		
\smallskip
		
The \textsc{combine} step reduces the number of bins by producing \textit{combined bins}, i.e. bins that are packed with alternating $MaxColor$ items from the M-bins and $OtherColors$ items from the top of F-bins. After the initial packing, there will be $F-P$ $OtherColors$ items remaining. Therefore, the \textsc{combine} step uses these items to yield $\big\lfloor{(F-P)/\lfloor(L-1)/2\rfloor} \big\rfloor$ combined bins. There will be $(F-P)\%\lfloor (L-1)/2\rfloor$ $OtherColor$ items remaining after the \textsc{combine} step. In the final packing, these remaining $OtherColor$ items will end up in a partially filled bin (note that this bin is not to be confused with the P-bin from the initial packing). Let $RO$ denote the number of these remaining items. We can pack $RO+1$ $MaxColor$ items in the partially filled bin. Note that if $RO = 0$, no such partially filled bin will exist in the final packing.
		
\smallskip

To summarize: after the initial packing we have $F$ full bins (F-bins) and $M$ single-item bins containing only one $MaxColor$ item (M-bins). The number of \textit{combined bins} will be:

\medskip

$C = ( \big\lfloor{(F-P)/\lfloor(L-1)/2\rfloor}\big\rfloor)$

\medskip

After \textsc{combine}, the number of $MaxColor$ items remaining is:

\medskip

$X = M-C \cdot \lceil(L-1)/2 \rceil $ if $RO =0$, and

\medskip

$X = M- C \cdot \lceil(L-1)/2 \rceil -RO-1$ if $RO >0$.


\medskip

So the total number of bins is at least $X+C$. If $R > 0$, then there is a partially full bin (that is not a P-bin) from the initial packing so there will be one additional bin in the final count. Similarly, if $RO > 0$, there is a partially full bin after \textsc{combine} so there will be one additional bin in the final count.

\medskip
		
\item $L$ is even and $Dispcrepancy \le 0$:

Since $Discrepancy$ is no more than zero, there must be enough $OtherColors$ items to pack all the $MaxColor$ items such that there will be no single-item bins containing only one $MaxColor$ item. Therefore, the number of bins depends on the bin capacity: specifically there will be $\lceil {n/L} \rceil$ bins.

\medskip

\item $L$ is odd and $Disrepancy > 0$:

We consider two sub-cases based on $Discrepancy$ ($D$):

\begin{itemize}

\item[$\bullet$] $D le \lceil{OtherCount / \lfloor{L/2}\rfloor}\rceil$:

In this case, there are enough $OtherColor$ items such that we can eventually reduce $D$ to zero. Therefore there will be no single-item bins containing only one $MaxColor$ item. Furthermore, we will pack $D$ bins before $D$ is reduced to zero. Therefore there are $n - D \cdot L$ items remaining after $D$ bins are packed in the initial packing and we will need $\lceil{n/L - D} \rceil$ bins for these items. Therefore the total number of bins is $D + \lceil{n/L - D} \rceil$.

\item[$\bullet$] {$D > \lceil{OtherCount / \lfloor{L/2}\rfloor}\rceil$}:

In this case, after the initial packing there will be some single-item bins that contain only one $MaxColor$ item. Since $L$ is odd and we pack each bin starting and ending with a $MaxColor$ item, each bin will fit one more $MaxColor$ item than $OtherColors$ item. Since in the initial packing we packed $\lceil{OtherCount / \lfloor{L/2}\rfloor}\rceil$ bins that contain both $MaxColor$ and $OtherColor$ items, we pack $\lceil{OtherCount / \lfloor{L/2}\rfloor}\rceil$ more $MaxColor$ items than $OtherColor$ items during the initial packing. Therefore the initial packing also yields $MaxCount - OtherCount - \lceil{OtherCount / \lfloor{L/2}\rfloor}\rceil$ $MaxColor$ items that are each packed in their own separate bin. Therefore, the total number of bins is $\lceil{OtherCount / \lfloor{L/2}\rfloor}\rceil + MaxCount - OtherCount - \lceil{OtherCount / \lfloor{L/2}\rfloor}\rceil = D$.

\end{itemize}

\item {$L$ is odd and $Discrepancy le 0$}

This case is identical to the case when $L$ is even and $Discrepancy \le 0$.

\end{enumerate}

\section{Conclusion}
\label{conc}

We studied the Colored Bin Packing problem where a set of items with varying weights and colors must be packed into bins of uniform weight limit such that no two items of the same color may be packed adjacently within a bin. We solved this problem for the case where there are two or more colors and we consider two settings: when the items have zero weight and when the items have unit weight. We presented optimal linear-time algorithms for both cases. Since our algorithms apply for two or more colors, they demonstrate that the problem does not get harder as the number of colors increases.

In our version of the problem, there is no ordering among the items. An important extension would be to develop algorithms that solve the problem when the items are received in a particular order and must be packed in that order. This setting is referred to as \textit{restricted offline}~\cite{jbalogh}. Another important extension would be to develop approximation algorithms for the setting where items have weight greater than one.


\begin{thebibliography}{4}

\footnotesize




\bibitem{jbalogh} J. Balogh, J. Bekesi, G. Dosa, H. Kellerer and Z. Tuza. Black and White Bin Packing.
\textit{ Approximation and Online Algorithms, 10th International Workshop, WAOA, Revised Selected Papers}, pp. 131 - 144, 2012.

\bibitem{mbohm} M. Bohm, J. Sgall and P. Vesely. Online Colored Bin Packing. \textit{Approximation and Online Algorithms, 12th International Workshop, WAOA 2014, Revised Selected Papers}, pp. 35 - 45, 2014.

\bibitem{mdawande} M. Dawande, J. Kalangnanam and J. Sethuraman. Variable Sized Bin Packing with Color Constraints. \textit{IElectronic Notes in Discrete Mathematics, Brazilian Symposium on Graphs, Algorithms and Combinatorics}, vol. 7, pp. 154-157, 2001.

\bibitem{gdosa}G. Dosa and L. Epstein. Colorful Bin Packing. \textit{Algorithm Theory, SWAT},  pp. 170-181, 2014.

\bibitem{lepstein} L. Epstein, C. Imreh, and A. Levin. Class Constrained Bin Packing Revisited. \textit{ Theoretical Computer Science}, vol. 411, no. 34, pp. 3073-3089, 2010.

\bibitem{kjansen} K. Jansen. An Approximation Scheme for Bin Packing with Conflicts. \textit{Journal of Combinatorial Optimization}, vol. 3, no. 4, pp. 363-377, 1999.

\bibitem{amuritiba} A. E. F. Muritiba, M. Iori, E. Malaguti, and P. Toth. Algorithms for the Bin Packing Problem with Conflicts. \textit{INFORMS Journal on Computing}, vol. 22, no. 3, pp. 401-415, 2010.


\bibitem{yoh} Y. Oh and  S. H. Son. On a Constrained Bin Packing Problem. \textit{Technical report CS-95-14. Department of Computer Science, University of Virginia, VA}, 1995.

\bibitem{sseidens} S. Seiden. On the Online Bin Packing Problem . \textit{Journal of the ACM (JACM) }, pp. 640-671, 2002.

\bibitem{exavier} E. Xavier and F. K. Miyazawa. The Class Constrained Bin Packing Problem with Applications to Video-on-Demand. \textit{Theoretical Computer Science}, vol. 393, no 1, pp. 240-259, 2008.



\end{thebibliography}
\end{document}